\documentclass[runningheads]{llncs}
\usepackage{amsmath}

\usepackage{amsthm,amssymb}
\usepackage{todonotes}
\usepackage[inline]{enumitem}
\usepackage{xspace}
\usepackage{subcaption}
\usepackage{comment}

\usepackage[bibliography=common]{apxproof}

\newtheoremrep{theorem2}[theorem]{Theorem}
\newtheoremrep{corollary2}[corollary]{Corollary}
\newtheoremrep{lemma2}[lemma]{Lemma}
\newtheoremrep{property2}[property]{Property}

\renewcommand{\orcidID}[1]{\href{https://orcid.org/#1}{\includegraphics[scale=.03]{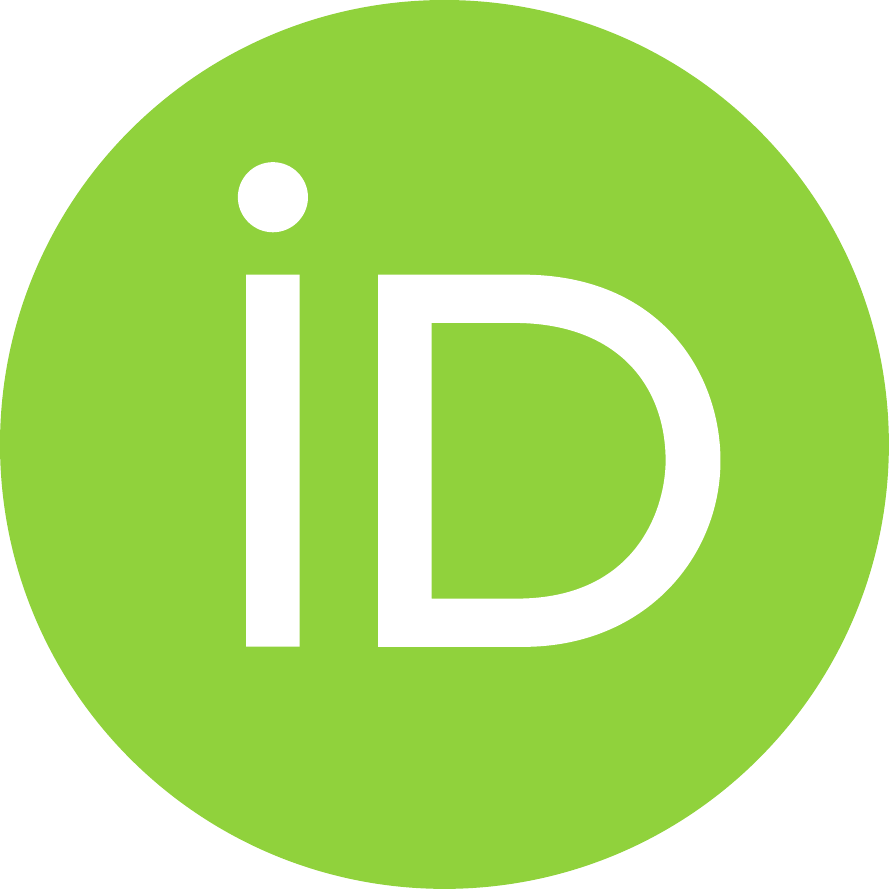}}}

\newcommand{\rique}{rique\xspace}
\newcommand{\Rique}{Rique\xspace}
\DeclareMathOperator{\sq}{riq}
\DeclareMathOperator{\deq}{deq}
\DeclareMathOperator{\D}{\texttt{D}}
\DeclareMathOperator{\St}{\texttt{S}}
\DeclareMathOperator{\Qu}{\texttt{Q}}
\DeclareMathOperator{\DEQ}{\texttt{DEQ}}
\DeclareMathOperator{\SQ}{\texttt{RIQ}}
\DeclareMathOperator{\pert}{pert}
\DeclareMathOperator{\skel}{skel}

\definecolor{defblue}{rgb}{0.121,0.47,0.705}
\definecolor{linkblue}{rgb}{0.098,0.098,0.4392}
\let\emph\relax
\DeclareTextFontCommand{\emph}{\color{defblue}\em}

\usepackage[colorlinks=true,linkcolor = linkblue,anchorcolor = linkblue,citecolor = linkblue,filecolor = linkblue,menucolor = linkblue,urlcolor = linkblue,bookmarks=true,bookmarksopen=true,bookmarksopenlevel=2,bookmarksnumbered=true,hyperindex=true,plainpages=false,pdfpagelabels=true,]{hyperref} 
\usepackage[capitalise,noabbrev,nameinlink]{cleveref}

\newtheorem{observation}[theorem]{Observation}
\Crefname{figure}{Fig.}{Figs.}
\title{The Rique-Number of Graphs\thanks{This work was initiated at the Bertinoro Workshop on Graph Drawing 2022.}}

\author{Michael~A.~Bekos\inst{1}\orcidID{0000-0002-3414-7444} \and Stefan~Felsner\inst{2}\orcidID{0000-0002-6150-1998} \and Philipp~Kindermann\inst{3}\orcidID{0000-0001-5764-7719} \and Stephen~Kobourov\inst{4}\orcidID{0000-0002-0477-2724} \and Jan Kratochv\'il\inst{5}\orcidID{0000-0002-2620-6133} \and Ignaz~Rutter\inst{6}\orcidID{0000-0002-3794-4406}}

\authorrunning{Bekos et al.}
 
\institute{%
Department of Mathematics, University of Ioannina, Ioannina, Greece\\ \email{bekos@uoi.gr}
\and
Institut f\"ur Mathematik, Technische Universit\"at Berlin, Berlin, Germany
\\\email{felsner@math.tu-berlin.de}
\and
Fachbereich IV - Informatikwissenschaften, Universit\"at Trier, Trier, Germany
\\\email{kindermann@uni-trier.de}
\and
Department of Computer Science, University of Arizona, Tucson, Arizona, USA\\
\email{kobourov@cs.arizona.edu}
\and
Department of Applied Mathematics, Charles University, Prague, Czech Republic\\
\email{honza@kam.mff.cuni.cz}
\and
Fakultät für Informatik und Mathematik, Universität Passau, Passau, Germany\\
\email{rutter@fim.uni-passau.de}
}

\begin{document}

\maketitle

\begin{abstract}
We continue the study of linear layouts of graphs in relation to known data structures. At a high level, given a data structure, the goal is to find a linear order of the vertices of the graph and a partition of its edges into pages, such that the edges in each page follow the restriction of the given data structure in the underlying order. In this regard, the most notable representatives are the stack and queue layouts, while there exists some work also for deques. 

In this paper, we study linear layouts of graphs that follow the restriction of a restricted-input queue (rique), in which insertions occur only at the head, and removals occur both at the head and the tail. We characterize the graphs admitting rique layouts with a single page and we use the characterization to derive a corresponding testing algorithm when the input graph is 
maximal planar. We finally give bounds on the number of needed pages (so-called rique-number) of complete~graphs.

\end{abstract}

\section{Introduction}

Linear graph layouts form an important methodological tool, since they provide a key-framework for defining different graph-parameters (including the well-known cutwidth~\cite{doi:10.1137/0125042}, bandwidth~\cite{DBLP:journals/jgt/ChinnCDG82}~and pathwidth~\cite{DBLP:journals/jct/RobertsonS83}). As a result, the corresponding literature is rather rich; see~\cite{DBLP:journals/eatcs/SernaT05}. Such layouts typically consist of an order of the vertices of a graph and an objective over its edges that one seeks to optimize. In the closely-related area of permutations and arrangements, back in~1973, Pratt~\cite{DBLP:conf/stoc/Pratt73} introduced and studied several variants of linear layouts that one can derive by leveraging different data structures to capture the order of the vertices (e.g., stacks, queues and deques). 

Formally, given $k$ data structures $\D_1,\ldots,\D_k$, a graph $G$ admits a $(\D_1,\ldots,\D_k)$-layout if there is a linear order $\prec$ of the vertices of $G$ and a partition of the~edges of $G$ into $k$ sets $E_1,\ldots,E_k$, called \emph{pages}, such that for each page $E_i$ in the partition, each edge $(u,v)$ of $E_i$ is processed by the data structure $\D_i$ by inserting $(u,v)$ to $\D_i$ at $u$ and removing it from $\D_i$ at $v$ if $u \prec v$ in the linear layout. If the sequence of insertions and removals is feasible, then $G$ is called a $(\D_1,\ldots,\D_k)$-graph. We denote the class of $(\D_1,\ldots,\D_k)$-graphs  by $\D_1+\ldots+\D_k$.
For a certain data structure $\D$, the $\D$-number of a graph $G$ is the smallest $k$ such that $G$ admits a $(\D_1,\ldots,\D_k)$-layout with $\D=\D_1=\ldots=\D_k$. This graph parameter has been the subject of intense research for certain data structures, as we discuss~below.

\begin{figure}[t!]
	\centering	
	\begin{subfigure}[b]{.32\textwidth}
		\centering
		\includegraphics[page=1]{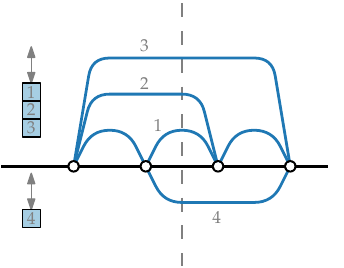}
		\caption{$(\St,\St)$-layout}
		\label{fig:stack}
	\end{subfigure}
	\hfil
	\begin{subfigure}[b]{.32\textwidth}
		\centering
		\includegraphics[page=2]{figures/layouts}
		\caption{$(\Qu,\Qu)$-layout}
		\label{fig:queue}
	\end{subfigure}
	\hfil
	\begin{subfigure}[b]{.32\textwidth}
		\centering
		\includegraphics[page=3]{figures/layouts}
		\caption{$\SQ$-layout}
		\label{fig:staq}
	\end{subfigure}
	\caption{
    Different linear layouts of the complete graph $K_4$.  The data structures are depicted in the states that corresponds to the dashed vertical line.}
	\label{fig:layouts}
\end{figure}

\begin{enumerate}
\item If $\D$ is a \emph{stack} (abbreviated by $\St$), then insertions and removals only occur at the head of $\D$; see \cref{fig:stack}. %
It is known that a non-planar graph may have linear stack-number, e.g., the stack-number of $K_n$ is $\lceil n/2 \rceil$~\cite{DBLP:journals/jct/BernhartK79}. A central result here is by Yannakakis, who back in 1986 showed that the stack-number of  planar graphs is at most $4$~\cite{DBLP:journals/jcss/Yannakakis89}, a bound which was only recently shown to be tight~\cite{DBLP:journals/jocg/KaufmannBKPRU20}. Certain subclasses of planar graphs, however, allow for stack-layouts with fewer than four stacks, e.g., see~\cite{DBLP:journals/algorithmica/BekosGR16,DBLP:journals/mp/CornuejolsNP83,DBLP:journals/dcg/FraysseixMP95,Ewald1973,DBLP:journals/dam/GuanY2019,DBLP:conf/focs/Heath84,DBLP:conf/esa/0001K19,DBLP:journals/appml/KainenO07,NC08,DBLP:conf/cocoon/RengarajanM95}.

\item If $\D$ is a \emph{queue} (abbreviated by $\Qu$), then insertions only occur at the head and removals only at the tail of $\D$; see \cref{fig:queue}. %
In this context, a breakthrough by Dujmovi\'c et al.~\cite{DBLP:journals/jacm/DujmovicJMMUW20} states that the queue-number of planar graphs is at most $49$, improving previous results~\cite{DBLP:journals/corr/BannisterDDEW18,DBLP:journals/siamcomp/BattistaFP13,DBLP:journals/jct/Dujmovic15,DBLP:journals/jgaa/DujmovicF18}. 
Even though this bound was recently improved to $42$~\cite{DBLP:conf/gd/BekosGR21}, the exact queue-number of planar graphs is not yet known; the current-best lower bound is $4$~\cite{DBLP:conf/gd/BekosGR21}. Again, 
several subclasses allow for layouts with significantly fewer than $42$ queues, e.g., see~\cite{DBLP:journals/algorithmica/AlamBGKP20,Ganley95,DBLP:journals/siamdm/HeathLR92,DBLP:conf/cocoon/RengarajanM95}.

\item If $\D$ is a \emph{double-ended queue} or \emph{deque} (abbreviated by $\DEQ$), then insertions and removals can occur both at the head and the tail of $\D$; we denote the deque-number of a graph $G$ by $\deq(G)$. This definition implies that $\St+\St \subseteq \DEQ \subseteq \St+\St+\Qu$.
A characterization by Auer et al.~\cite{DBLP:journals/jgaa/AuerBBBG18} (stating that a graph has deque-number 1 if and only if it is a spanning subgraph of a planar graph with a Hamiltonian path) implies that the first containment is strict, because a maximal planar graph with a Hamiltonian path but not a Hamiltonian cycle (e.g., the Goldner-Harary graph~\cite{GH75}) admits a $\DEQ$-layout, but not an $(\St,\St)$-layout. The second containment is also strict because $(\St,\St,\Qu)$-graphs can be non-planar (e.g., $K_6$~\cite{DBLP:journals/corr/abs-2107-04993}). Hence, $\St+\St \subsetneq \DEQ \subsetneq \St+\St+\Qu$ holds.
\end{enumerate}

\noindent\textbf{Our contribution.} In this work, we focus on the case where the data structure~$\D$ is a \emph{restricted-input queue} or \emph{\rique} (abbreviated by $\SQ$), in which insertions occur only at the head, and removals occur both at the head and the tail of $\D$; see~\cref{fig:staq}. We first characterize the graphs with $\rique$-number~$1$ as those admitting a planar embedding with a so-called \emph{strongly 1-sided subhamiltonian path}, that is, a Hamiltonian path $v_1,\ldots,v_n$ in some plane extension of the embedding such that each edge $(v_i,v_j)$ with $1 < i < j \le n$ leaves $v_i$ on the same side of the path; see~\cref{fig:1sided-hamiltonian}. This characterization allows us to derive an inclusion relationship similar to the one above for deques (namely, $\St,\Qu \subsetneq \SQ \subsetneq \St+\Qu$; see \cref{obs:inclusions}) and corresponding recognition algorithms for graphs with \rique-number~1 under some assumptions (\cref{thm:plane-staq-ham}). Then, we focus on bounds on~the \rique-number of a graph $G$, which we denote by $\sq(G)$. Our contribution is an edge-density bound for the graphs with \rique-number~$k$ (\cref{thm:density}), and a lower and an upper bound on the \rique-number of complete graphs (\cref{thm:staq-bounds}).

\begin{figure}[t]
    \centering
    \includegraphics[page=2]{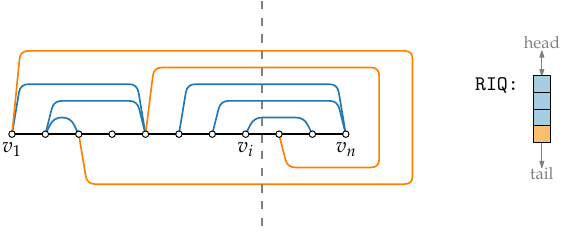}
    \caption{A strongly 1-sided Hamiltonian path and the state of the $\SQ$ that processes it right after processing the edges incident to~$v_i$.}
    \label{fig:1sided-hamiltonian}
\end{figure}

\section{Preliminaries}
\label{sec:preliminaries}

We start with definitions that are central in \cref{sec:staq-1}.
Given a \rique-layout, we call an edge $(u,v)$ a \emph{head-edge} (\emph{tail-edge}), if $(u,v)$ is removed at $v$ from the head (tail) of the $\SQ$. 
A \emph{strongly 1-sided Hamiltonian path} of a plane graph is a Hamiltonian path $v_1,\ldots,v_n$  such that each~edge $(v_i,v_j)$ with $1<i<j \le n$ leaves $v_i$ on the same side of the path, say w.l.o.g.\ the \emph{left} one, i.e., between $(v_{i-1},v_i)$ and $(v_i,v_{i+1})$ in clockwise order around $v_i$ (see \cref{fig:1sided-hamiltonian}).
A plane graph is \emph{strongly 1-sided Hamiltonian} if it contains a strongly 1-sided Hamiltonian path.
A planar graph is \emph{strongly 1-sided Hamiltonian} if it admits a planar embedding that contains a strongly 1-sided Hamiltonian path.
A planar (plane) graph $G$ is \emph{strongly 1-sided subhamiltonian} if there exists a planar (plane) supergraph $H$ of $G$ that is strongly 1-sided Hamiltonian.

Another key-tool that we leverage in \cref{sec:recognition} is the SPQR-tree. This data structure, introduced by Di Battista and Tamassia~\cite{dt-ipt-89,dt-ogasp-90}, compactly represents~all planar embeddings of a biconnected planar graph; see \cref{fig:spqr} for an example. It is unique and can be computed in linear time~\cite{gm-ltisp-00}. We assume familiarity with SPQR-trees; for a brief introduction refer to \cref{app:preliminaries}.

\begin{figure}[t]
    \centering
    \includegraphics{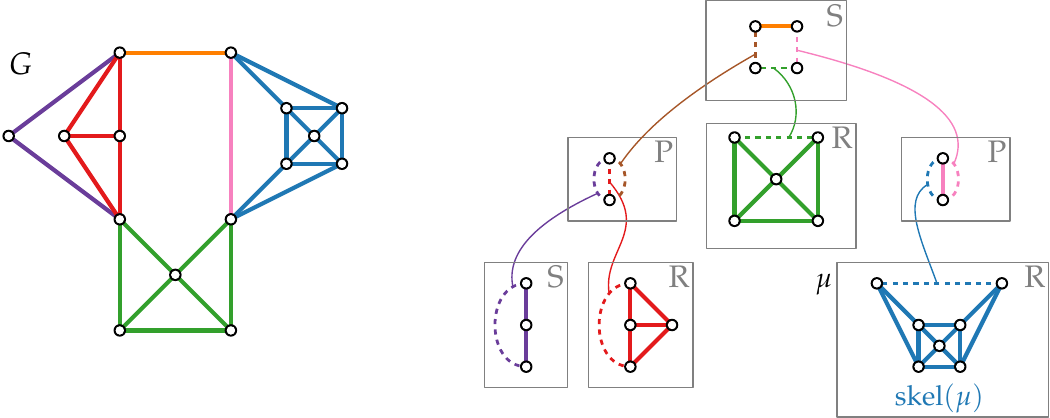}
    \caption{An SPQR-tree, omitting the Q-nodes.}
    \label{fig:spqr}
\end{figure}

\begin{toappendix}
\section{Omitted Material from \cref{sec:preliminaries}}
\label{app:preliminaries}
The SPQR-tree~$\mathcal T$ of a biconnected graph $G=(V,E)$ is a tree such that each node~$\mu$ of~$\mathcal T$ is associated with a multigraph~$\skel(\mu)$, called the \emph{skeleton} of~$\mu$, and whose edges can be~either \emph{real edges} or \emph{virtual edges}, and whose vertex sets are subsets of~$V$.  The leaves~of the SPQR-tree are Q-nodes, whose skeleton consists of two vertices that are~joined by one real edge and by one virtual edge that represents the rest of the graph.  The internal nodes of~$\mathcal T$ are either S-, P- or R-nodes, depending on their skeleton.  The skeleton of an S-node is a simple cycle of virtual edges, the skeleton of a P-node consists of two vertices that are joined by at least three parallel virtual edges, and the skeleton of an R-node is a triconnected graph consisting of virtual edges.  Each edge~$\mu\nu$ of~$\mathcal T$ is associated with two virtual edges, one in~$\skel(\mu)$ and one in~$\skel(\nu)$ that have the same endpoints in such a way that we obtain for each node~$\mu$ a bijection between the virtual edges of~$\skel(\mu)$ and the edges of~$\mathcal T$ incident to~$\mu$.  Further no two S-nodes and no two P-nodes are adjacent in~$\mathcal T$.

For an edge~$\mu\nu$ of~$\mathcal T$, the skeletons $\skel(\mu)$ and~$\skel(\nu)$ can be \emph{joined} by identifying~the endpoints of the virtual edges corresponding to the edge~$\mu\nu$ and then removing the virtual edge that corresponds to~$\mu\nu$.  The key property of the SPQR-tree of $G$ is that by joining along all tree edges, we obtain the graph~$G$.  If each skeleton is further equipped with a planar embedding, then joining them yields a planar embedding of~$G$, and in fact all planar embeddings of~$G$ can be obtained in this way.  Note that, due to the restricted nature of the skeletons, the embedding choices for skeletons are limited: S- and Q-nodes have a unique planar embedding, whereas in a P-node we can arbitrarily permute the order of the virtual edges and an R-node has a unique planar embedding up to reflection.
Rooting the SPQR-tree at a Q-node corresponding to some reference edge~$st$, defines for each non-root node~$\mu$ a unique \emph{parent edge} in~$\skel(\mu)$, namely the virtual edge that corresponds to the edge connecting~$\mu$ to its parent.  We call the endpoints of the parent edge of~$\skel(\mu)$ the poles of~$\mu$.  If we additionally restrict the choices of the planar embeddings of the skeletons so that the parent edge is incident to the outer face, the SPQR-tree represents exactly those planar embeddings of $G$ where the reference edge~$st$ is incident to the outer face.  The pertinent graph~$\pert(\mu)$ of a node~$\mu$ is the subgraph represented by~$\mu$ and all its children; it can be obtained by joining~$\mu$ and all nodes below~$\mu$.  The pertinent graph of the root is~$G$ itself.
\end{toappendix}

\section{Characterization of Graphs with \Rique-Number 1}
\label{sec:staq-1}

In this section, we discuss properties of graphs with \rique-number~1. We first characterize these graphs in \cref{lem:pattern-staq} in terms of the following forbidden~pattern.

\begin{enumerate}[label={P.\arabic*}]
\item \label{pattern-staq} 
Three edges $\langle e_a,e_b,e_c\rangle$ with $e_a=(a,a')$, $e_b=(b,b')$ and $e_c=(c,c')$ form Pattern~\ref{pattern-staq} in a linear layout if and only if $a \prec b \prec c \prec b' \prec \{a',c'\}$;~see~\cref{fig:pattern-staq}.
\end{enumerate}

\begin{figure}[b]
    \centering
    \includegraphics{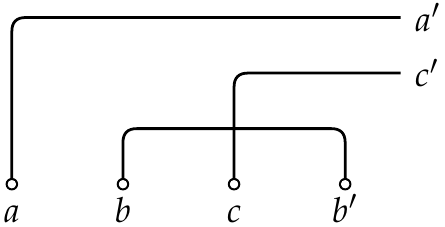}
    \caption{Forbidden Pattern~\ref{pattern-staq}}
    \label{fig:pattern-staq}
\end{figure}

\begin{lemma}\label{lem:pattern-staq}
  A graph has \rique-number 1 if and only if it admits a linear order avoiding Pattern~\ref{pattern-staq}.
\end{lemma}
\begin{proof}
Let $G$ be a graph with \rique-number 1 and assume for a contradiction that a linear order of it contains Pattern~\ref{pattern-staq}. 
The edge $e_a$ is inserted into data structure $\SQ$ before the edge $e_b$ is inserted, but removed after $e_b$ is removed. Hence, $e_b$ cannot be removed at the tail of $\SQ$, so it has to be removed at its head.
However, the edge $e_c$ is inserted after the edge $e_b$ is inserted, but also removed after $e_b$ is removed, so $e_b$ also cannot be removed at the head; a contradiction.

For the other direction, assume that $G$ has \rique-number greater than 1.~We will prove that every linear order of $G$ contains Pattern~\ref{pattern-staq}. Let  $\prec$ be such an order. Since $G$ has \rique-number greater than 1 and all insertions into a $\SQ$ happen on the same side, at some time $b'$ there is an edge $e_b$ to be removed that is neither at the head nor at the tail of $\SQ$. Since $e_b$ is not~at~the head, there is some other edge $e_c$ that was inserted into $\SQ$ after $e_b$ and is still there at time~$b'$.
Since $e_b$ is not at the tail, there is some other edge $e_a$ that was inserted into $\SQ$ before $e_b$ and still is  there. Then, $\langle e_a,e_b,e_c\rangle$ form Pattern~\ref{pattern-staq}.
\end{proof}

\noindent We are now ready to completely characterize the graphs with \rique-number~$1$. 

\begin{theorem}\label{thm:characterization}
A graph $G$ has \rique-number 1 if and only if~$G$ is planar strongly 1-sided subhamiltonian.
\end{theorem}
\begin{proof}
First, assume that $G$ can be embedded so that it contains a strongly 1-sided subhamiltonian path $v_1,\ldots,v_n$. For a contradiction, assume further that $\langle e_a=(a,a'),e_b=(b,b'),e_c=(c,c')\rangle$ form Pattern~\ref{pattern-staq} in the order $v_1,\ldots,v_n$.
Note that $e_a$, $e_b$, and $e_c$ leave $a$, $b$, and $c$ on the left side, respectively.
If $e_b$ enters $b'$ from the left, then $e_b$ crosses~$e_c$ as $b\prec c\prec b'\prec c'$. So, $e_b$ has to enter $b'$ from the right. Then, however, $e_b$~crosses $e_a$ since $a\prec b\prec b'\prec a'$; a contradiction. So, %
by \cref{lem:pattern-staq}, $G$ has~\rique-number~1.
  
Assume now that $G$ has \rique-number 1. By \cref{lem:pattern-staq}, $G$ admits a linear order $v_1,\ldots,v_n$ avoiding Pattern~\ref{pattern-staq}. 
W.l.o.g.\ we assume that $G$ contains all edges in $\{(v_1,v_2),\ldots,(v_{n-1},v_n)\}$ and prove that $G$ is strongly 1-sided Hamiltonian.

Consider a vertex $v_i$. We order the edges around $v_i$ counter-clockwise as follows; see \cref{fig:order-around-vertex}.
  \begin{enumerate*}[label=(\roman*)]
      \item The edge $(v_i,v_{i+1})$ (for $i<n$);
      \item the outgoing head-edges of $v_i$, ordered in increasing order by the
        index of the target vertex;
      \item the outgoing tail-edges of $v_i$, ordered in decreasing order by the
        index of the target vertex;
      \item the incoming head-edges of $v_i$, ordered in increasing order by the
        index of the source vertex;
      \item the edge $(v_{i-1},v_i)$ (for $i>1$);
      \item the incoming tail-edges of $v_i$, ordered in increasing order by the
        index of the source vertex. 
  \end{enumerate*}
  This ensures that all edges leave $v_i$ on the correct side of the Hamiltonian path. It remains to be shown that this embedding is plane. To this end, assume that there are two edges $(v_i,v_j)$ and $(v_k,v_\ell)$ that cross. W.l.o.g.\ we assume that $i<k$. %
  
\begin{figure}[t]
    \centering
    \includegraphics{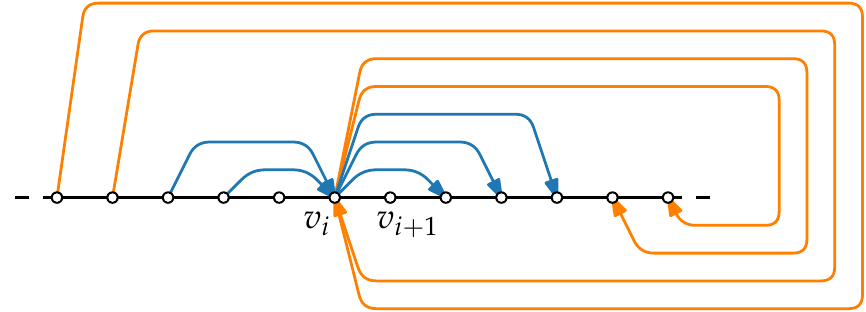}
    \caption{Ordering of the edges around a vertex $v_i$.}
    \label{fig:order-around-vertex}
\end{figure}
  
  If $(v_k,v_\ell)$ is a head-edge, then it leaves and enters $v_k$ and $v_\ell$ on the same side of the Hamiltonian path as $(v_i,v_j)$ leaves $v_i$. Hence,
  $(v_i,v_j)$ and $(v_k,v_\ell)$ cross only if $(v_i,v_j)$ also enters $v_j$ on the same side. So, $(v_i,v_j)$ is also a head-edge with $i<k<j<\ell$. However, since $(v_i,v_j)$ entered $\SQ$ at the head before $(v_k,v_\ell)$, it cannot leave $\SQ$ at the head before $(v_k,v_\ell)$; a contradiction.
  
  If $(v_k,v_\ell)$ is a tail-edge, then $(v_i,v_j)$ leaves $v_i$ on the same side of the Hamiltonian path as $(v_k,v_\ell)$ leaves $v_k$, but $(v_k,v_\ell)$ enters $v_\ell$ on the other side. 
  If $(v_i,v_j)$ is a head-edge, then we must have $i<k<j$. However, since $(v_i,v_j)$ entered $\SQ$ at the head before $(v_k,v_\ell)$, it cannot leave $\SQ$ at the head before $(v_k,v_\ell)$; a contradiction.
  Otherwise $(v_i,v_j)$ is a tail-edge, and we must have $i<k<\ell<j$. However, since $(v_i,v_j)$ entered $\SQ$ at the head before $(v_k,v_\ell)$, it cannot leave $\SQ$ at the tail after $(v_k,v_\ell)$; a contradiction.
 
It follows that no two edges cross, as desired. This concludes the proof.  
\end{proof}

The definition of a \rique implies $\texttt{X} \subseteq \SQ \subseteq \St+\Qu$, where $\texttt{X} \in \{\St,\Qu\}$. By~\cref{thm:characterization}, both inclusions are strict, as $K_4 \in \SQ$ (see \cref{fig:staq}) but it admits neither a stack-layout (since it is not outerplanar~\cite{DBLP:journals/jct/BernhartK79}) nor a queue-layout (since any linear order yields a 2-rainbow~\cite{DBLP:journals/siamcomp/HeathR92}), and $K_{6}$ admits an $(\St,\Qu)$-layout~\cite{DBLP:journals/corr/abs-2107-04993} but is not planar and therefore $K_6 \notin \SQ$.%
\begin{observation}\label{obs:inclusions}
\normalfont $\texttt{X} \subsetneq \SQ \subsetneq \St+\Qu$, where $\texttt{X} \in \{\St,\Qu\}$
\end{observation}

\section{Recognition of graphs with \Rique-Number 1}
\label{sec:recognition}

With the characterization of \cref{thm:characterization} at hand, we now turn our focus to the recognition problem, where we present two algorithms:   
\begin{enumerate*}[label=(\roman*)]
\item the first one is simple and tests whether a plane graph is strongly $1$-sided Hamiltonian, while
\item the second one is more elaborate and tests whether a planar graph is strongly $1$-sided Hamiltonian.  
\end{enumerate*}
Even though our algorithms do not solve the general case of testing whether a graph has \rique-number $1$ (or equivalently by \cref{thm:characterization} whether it is strongly $1$-sided subhamiltonian), they can be leveraged for testing, e.g., whether a maximal planar graph or a $3$-connected planar graph has \rique-number~$1$. 

\begin{theorem}\label{thm:plane-staq-ham}
Given a plane $n$-vertex graph $G$, there is an $O(n^2)$-time algorithm to test whether $G$ is plane strongly 1-sided Hamiltonian.
\end{theorem}
\begin{proof}
After guessing the first edge of the path, for which there are $O(n)$ choices, we assume that we have computed a subpath~$\pi = v_1,\dots,v_i$, $2 \le i < n$ of a strongly 1-sided Hamiltonian path of $G$. We claim that the next vertex~$v_{i+1}$ is uniquely determined by~$\pi$.  Consider the edges of $G$ incident to~$v_{i}$ in counterclockwise order, starting from the edge after~$(v_{i-1},v_i)$.  Let~$e$ be the first edge in this order, whose other endpoint does not lie on~$\pi$.  We choose this endpoint as~$v_{i+1}$.  This is correct, since choosing an endpoint of an edge preceding~$e$ visits a vertex twice, whereas choosing an endpoint of an edge succeeding~$e$ would imply that $e$ leaves the resulting path on the wrong side.
The above argument shows that, after guessing an initial edge, the remainder of the 1-sided Hamiltonian path is uniquely defined, if it exists.  Since a single starting edge can be tested in $O(n)$ time, the overall time complexity of our algorithm is $O(n^2)$.
\end{proof}

\begin{corollary}
Given a maximal planar graph $G$ with $n$ vertices, there is an $O(n^2)$-time algorithm to test whether $G$ has \rique-number~$1$.
\end{corollary}

\begin{figure}[t]
    \centering
  \begin{subfigure}{.47\linewidth}
    \includegraphics[page=1]{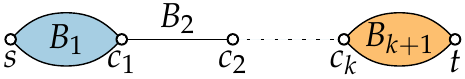}
    \caption{}
  \end{subfigure}
  \hfill
  \begin{subfigure}{.47\linewidth}
    \includegraphics[page=2]{figures/blockcut.pdf}
    \caption{}
  \end{subfigure}
  \begin{subfigure}{.47\linewidth}
    \includegraphics[page=3]{figures/blockcut.pdf}
    \caption{}
  \end{subfigure}
    \caption{(a) A block-cut tree;
    (b) a strongly 1-sided Hamiltonian embedding for each block; 
    (c) a strongly 1-sided Hamiltonian embedding for the whole graph.}
    \label{fig:blockcut}
\end{figure}

\begin{theorem}
Given a planar $n$-vertex graph $G$, there is an $O(n^4)$-time algorithm to test whether $G$ is planar strongly 1-sided Hamiltonian.
\end{theorem}
\begin{proof}
To prove the statement, we assume that the endpoints $s,t$ of the Hamiltonian path are specified as part of the input and we show that testing whether~$G$ admits a planar embedding containing a strongly 1-sided Hamiltonian $st$-path can be done in $O(n^2)$ time. In the positive case, we say that~$G$ is \emph{$st$-1-sided}.
    
If~$G$ is not biconnected, then for $G$ to be $st$-1-sided its block-cut tree must~be a path~$B_1,c_1,B_2,\dots,\allowbreak c_k, B_{k+1}$, such that $s \in B_1$ and~$t \in B_{k+1}$ (or vice versa; here~$k$ denotes the number of cutvertices of $G$).  We set~$c_0 = s$ and~$c_{k+1} = t$ and claim that~$G$ is $st$-1-sided if and only if each block~$B_i$ is $c_{i-1}c_{i}$-1-sided for~$i=1,\dots,k+1$.
The necessity is clear, we prove the sufficiency.   Let~$\mathcal E_i$ be a planar embedding of~$B_i$ containing a strongly 1-sided Hamiltonian $c_{i-1}c_i$-path~$p_i$ for~$i=1,\dots,k+1$.  We modify the embedding~$\mathcal E_i$ such that the first edge of~$p_i$ lies on the outer face, and combine $\mathcal E_{i-1}$ and $\mathcal E_i$ in such a way that the first edge of~$p_i$ follows the last edge of~$p_{i-1}$ in counterclockwise order around~$c_i$.  Then the path~$p$ obtained by concatenating~$p_i$, $i=1,\dots,k+1$ is a strongly 1-sided Hamiltonian path in the resulting embedding~$\mathcal E$ of~$G$; see \cref{fig:blockcut}.
    
Hence, we may assume that~$G$ consists of a single block.  Since the case where $G$ consists of a single edge can be handled trivially, we focus on the case where~$G$ is biconnected.  To determine whether~$G$ is $st$-1-sided, we use a dynamic program based on an SPQR-tree $\mathcal{T}$ of $G$.  We root $\mathcal{T}$ at an edge incident to~$t$ and for each node~$\mu$ of $\mathcal{T}$ with poles~$u,v$, we want to answer the following questions:  If~$s \notin \pert(\mu)$, we want to know for each of the two ordered pairs of poles~$(x,y) \in \{(u,v), (v,u)\}$ whether~$\pert(\mu)$ has an embedding with~$x,y$ on the outer face such that it contains a strongly 1-sided Hamiltonian path from~$x$ to~$y$ that starts with the edge that follows the parent edge counterclockwise around~$x$;  in the positive case.  We define the set $L(\mu)$ as those ordered pairs~$(x,y)$ where this is the case.  For a pair~$(x,y) \in L(\mu)$, we denote by~$\mathcal E_\mu(x,y)$ the corresponding embedding of~$\pert(\mu)$ and by~$P_\mu(x,y)$ the corresponding path. 
If~$s \in \pert(\mu)$, then for each~$x \in \{u,v\}$ and~$Y \subseteq \{u,v\} \setminus \{x\}$ we want to know whether~$\pert(\mu)$ has an embedding~$\mathcal E_\mu(x,Y)$ such that~$u,v$ are incident to the outer face and there is a strongly 1-sided path~$P_\mu(x,Y)$ from~$s$ to~$x$ that visits all vertices of~$\pert(\mu) - Y$. As above, for node~$\mu$, we define~$L(\mu)$ as the set of all pairs~$(x,Y)$ where this is possible.
   
Consider the root $r$ of $\mathcal{T}$ and let~$\mu$ be its child with poles~$u,t$.  Then $G$ is $st$-1-sided if and only if and only if $(t,\emptyset) \in L(\mu)$.  The necessity is clear.  For the sufficiency, observe that~$P_\mu(t,\emptyset)$ is a strongly 1-sided $st$-path in the embedding of~$G$ obtained from~$\mathcal E_\mu(t,\emptyset)$ by adding the edge~$ut$ in the outer face.
We compute the set~$L(\mu)$ for each node $\mu$ of $\mathcal{T}$ (together with corresponding embeddings of~$\pert(\mu)$ and paths) by a bottom-up traversal of~$\mathcal T$ as follows. Let $\mu$ be a node of~$\mathcal{T}$ in this traversal with poles~$u,v$. If $\mu$ is not a leaf in $\mathcal{T}$, we denote by $\mu_1,\dots,\mu_k$ its children, and we assume that~$L(\mu_i)$ has already been computed for~$i=1,\dots,k$. We next distinguish cases based on the type of~$\mu$.
    
    \begin{figure}[t]
        \centering
        \begin{subfigure}[b]{.4\textwidth}
        \centering
            \includegraphics[page=1]{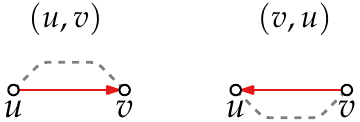}
            \caption{$u\neq s\neq v$}
            \label{fig:q-node-1}
        \end{subfigure}
        \hfill
        \begin{subfigure}[b]{.55\textwidth}
        \centering
            \includegraphics[page=2]{figures/q-node.pdf}
            \caption{$v=s$}
            \label{fig:q-node-2}
        \end{subfigure}
        \caption{The tuples of~$L(\mu)$ for a Q-node;  the corresponding paths are red.}
        \label{fig:q-node}
    \end{figure}
    
\smallskip\noindent\textbf{Case 1: $\mu$ is a Q-node.}    If~$u \ne s \ne v$, then~$L(\mu) = \{(u,v),(v,u)\}$. And for~$(x,y) \in L(\mu)$ the embedding~$\mathcal E_\mu(x,y)$ and the path~$P_\mu(x,y)$ are trivial; see \cref{fig:q-node-1}.  Otherwise, assume w.l.o.g.~$s = v$. Then~$L(\mu) = \{(v,\{u\}),(u,\emptyset)\}$.  Again for~$(x,Y) \in L(\mu)$, $\mathcal E_\mu(x,Y)$ and~$P_\mu(x,Y)$ can be defined trivially; see \cref{fig:q-node-2}.

    \begin{figure}[t]
        \centering
        \begin{subfigure}[b]{.45\textwidth}
        \centering
        \includegraphics[page=1]{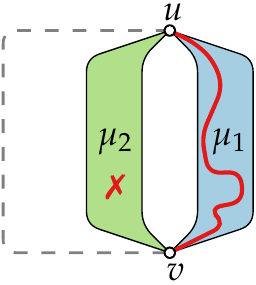}
        \hfill
        \includegraphics[page=2]{figures/p-node}
        \caption{$s\notin\pert(\mu)$ or $s=v$}
        \label{fig:p-node-1}
        \end{subfigure}
        \hfill
        \begin{subfigure}[b]{.45\textwidth}
        \centering
        \includegraphics[page=4]{figures/p-node}
        \hfill
        \includegraphics[page=3]{figures/p-node}
        \caption{$s\notin\{u,v\}, Y=\{v\}$}
        \label{fig:p-node-2}
        \end{subfigure}
        
        \medskip
        \begin{subfigure}[b]{\textwidth}
        \centering
        \includegraphics[page=6]{figures/p-node}
        \hfil
        \includegraphics[page=5]{figures/p-node}
        \hfil
        \includegraphics[page=7]{figures/p-node}
        \caption{$s\notin\{u,v\}, Y=\emptyset$}
        \label{fig:p-node-3}
        \end{subfigure}
        \caption{Paths~$P(v,u)$, $P(u,Y)$ in a P-node.  The vertices in~$Y$ are black.}
        \label{fig:p-node}
    \end{figure}
       
\smallskip\noindent\textbf{Case 2: $\mu$ is a P-node.} Assume first that $s\notin\pert(\mu)$; see \cref{fig:p-node-1}.  We show how to test whether~$(v,u) \in L(\mu)$. The case of $(u,v)$ is symmetric.  First~$(v,u) \in L(\mu)$ requires~$k=2$ and that only one of the children, say~$\mu_1$, is not a Q-node. If so, $(v,u) \in L(\mu)$ if and only if~$(v,u) \in L(\mu_1)$. Also, $P_\mu(v,u) = P_{\mu_1}(v,u)$~and $\mathcal E_\mu(v,u)$ is obtained by embedding the edge represented by~$\mu_2$ to the left of~$\mathcal E_{\mu_1}(v,u)$.
    
Now, consider the case that $s\in\pert(\mu)$. Assume first that~$s$ is a pole of $\mu$; see \cref{fig:p-node-1}. Then, any $1$-sided path of $\mu$ unavoidably visits the other pole. In fact, only a single child can be traversed, i.e., $k=2$, and one child, say~$\mu_2$, is a Q-node. If this is not the case, $L(\mu) = \emptyset$. Otherwise,~$L(\mu) = L(\mu_1)$.   For~$(x,Y) \in L(\mu_1)$, we set~$p_\mu(x,Y)=p_{\mu_1}(x,Y)$ and we define~$\mathcal E_\mu(x,Y)$ as the embedding obtained from~$\mathcal E_{\mu_1}(x,Y)$ by putting the edge represented by~$\mu_2$ to its left parallel to it.
    
Assume now that~$s$ is not a pole and it lies, w.l.o.g., in~$\pert(\mu_1)$.  Let~$(x,Y)$ be a pair with~$x \in \{u,v\}$, $Y \subseteq \{u,v\}\setminus\{x\}$.  W.l.o.g.\ we assume~$x=u$.  The case~$x=v$ is analogous.  Then either~$Y= \{v\}$ or~$Y = \emptyset$.
If~$v \in Y$ (see \cref{fig:p-node-2}), then~$(u,Y) \in L(\mu)$ if and only if~$(u,Y) \in L(\mu_1)$ and~$k=2$ and $\mu_2$ is a Q-node.  In that case, we set~$P_\mu(u,Y) = P_{\mu_1}(u,Y)$ and we define~$\mathcal E_\mu(u,Y)$ as the embedding obtained from~$\mathcal E_{\mu_1}(x,Y)$ by embedding the edge represented by~$\mu_2$ to its left.
If~$Y = \emptyset$ (see \cref{fig:p-node-3}), then we distinguish cases based on whether there is a second child, say~$\mu_2$, that is not a Q-node.  If there is none, then~$\mu_2$ is a Q-node and then~$(u,\emptyset) \in L(\mu)$ if and only if either $(u,\emptyset) \in L(\mu_1)$ or if~$(v,\{u\}) \in L(\mu_1)$.  In these cases, we set~$P_\mu(u,\emptyset) = P_{\mu_1}(u,\emptyset)$ or~$P_\mu(u,\emptyset) = P_{\mu_1}(v,\{u\}) \cdot (v,u)$.  The embedding~$\mathcal E_\mu(u,\emptyset)$ is obtained by embedding the edge represented by~$\mu_2$ on the left side of~$\mathcal E_{\mu_1}(u,\emptyset)$ or $\mathcal E_{\mu_1}(v,\{u\})$, respectively.
Otherwise~$\mu_2$ is not a Q-node.  It is then necessary that~$k \le 3$ and if~$\mu_3$ exists, it must be a Q-node.
Now,~$(u,\emptyset) \in L(\mu)$ if and only if~$(v,\{u\}) \in L(\mu_1)$ and~$(v,u) \in L(\mu_2)$.  In this case, we define~$P_\mu(u,\emptyset) = P_{\mu_1}(v,\{u\}) \cdot P_{\mu_2}(v,u)$ and the embedding $\mathcal E_\mu(u,\emptyset)$ is obtained by embedding~$\mathcal E_{\mu_2}(v,u)$ to the left of~$\mathcal E_{\mu_1}(v,\{u\})$ and the edge represented by~$\mu_3$, if it exists, to the left of that.
    
    \begin{figure}[t]
        \centering
        \begin{subfigure}[b]{.18\textwidth}
            \centering
            \includegraphics[page=1]{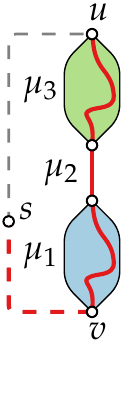}
            \caption{$s\notin\pert(\mu)$}
            \label{fig:s-node-1}
        \end{subfigure}
        \begin{subfigure}[b]{.15\textwidth}
            \centering
            \includegraphics[page=2]{figures/s-node}
            \caption{$s\in \mu_{i>2}$}
            \label{fig:s-node-2}
        \end{subfigure}
        \begin{subfigure}[b]{.4\textwidth}
            \centering
            \includegraphics[page=3]{figures/s-node}
            \includegraphics[page=4]{figures/s-node}
            \includegraphics[page=5]{figures/s-node}
            \caption{$s\in \mu_{2}$}
            \label{fig:s-node-3}
        \end{subfigure}
        \hfill
        \begin{subfigure}[b]{.23\textwidth}
            \centering
            \includegraphics[page=6]{figures/s-node}
            \includegraphics[page=7]{figures/s-node}
            \caption{$s\in \mu_{1}$}
            \label{fig:s-node-4}
        \end{subfigure}
        \caption{S-node}
        \label{fig:s-node}
    \end{figure}
    
    \smallskip\noindent\textbf{Case 3: $\mu$ is an S-node.}  
Let the children of $\mu$ be numbered so that~$v$ is a pole of~$\mu_1$.  Further, we denote by~$v_i$ the pole shared by~$\mu_i$ and~$\mu_{i+1}$ for~$i=1,\dots,k-1$. To ease the presentation, we also write~$v_0 = v$ and~$v_{k+1} = u$.
        
We start with the case that $s\notin\pert(\mu)$; see~\cref{fig:s-node-1}. We show how to test whether~$(v,u) \in L(\mu)$; the case of~$(u,v)$ is analogous. Then~$(v,u) \in L(\mu)$ if~and only if~$(v_{i-1},v_{i}) \in L(\mu_i)$ for~$i=1,\dots,k$.  In that case,~$P_\mu(v,u)$ is obtained by concatenating~$P_{\mu_i}(v_{i-1},v_i)$ for~$i=1,\dots,k$, while~$\mathcal E_\mu(v,u)$ is obtained by merging~$\mathcal E_{\mu_i}(v_{i-1},v_i)$ for~$i=1,\dots,k$.

Now, consider the case that~$s\in\pert(\mu)$.
Consider a pair~$(x,Y)$ as above.  We show the case~$x=u$, the case $x=v$ can be handled analogously.  If~$s=v$, then we cannot avoid visiting $s$, and we proceed as in the case of~$(v,u)$ where~$s$ is not in $\pert(\mu)$.  
Now consider the case that~$s$ is not a pole; see~\cref{fig:s-node-2,fig:s-node-3,fig:s-node-4}.  Let~$i$ be the smallest index so that~$s$ belongs to~$\pert(\mu_i)$ (observe that $s$ belongs to more than one pertinent graphs if and only if it is a vertex of~$\skel(\mu)$).  If~$i>2$, then $L(\mu) = \emptyset$, i.e., there is no path from~$s$ to~$x$ that visits $v_1$; see \cref{fig:s-node-2}.  Similarly, for $i=2$ we have~$(u,Y) \in L(\mu)$ if and only if~$\mu_1$ is a Q-node, $Y = \{v\}$, $(v_2,\emptyset) \in L(\mu_2)$, and~$(v_{j-1},v_{j}) \in L(\mu_j)$ for~$j=3,\dots,k$; see \cref{fig:s-node-3}. In this case, $P_\mu(x,Y)$ is composed by concatenating $P_{\mu_2}(v_2,\emptyset)$ with $P_{\mu_j}(v_{j-1},v_j)$ for $j=3,\ldots,k$, while the embedding $\mathcal E_\mu(x,Y)$ is obtained by merging the edge representing $\mu_1$ with $\mathcal E_{\mu_2}(v_2,\emptyset)$ with the embeddings of $\mathcal E_{\mu_j}(v_{j-1},v_{j})$ for $j=3,\ldots,k$
If~$i=1$, $(u,Y) \in L(\mu)$ if and only if~$(v_1,Y) \in L(\mu_1)$ and~$(v_{j-1},v_j) \in L(\mu_j)$ for $j=2,\dots,k$; see \cref{fig:s-node-4}.  In this case $P_\mu(x,Y)$ is composed by concatenating $P_{\mu_1}(v_1,Y)$ with~$P_{\mu_j}(v_{j-1},v_j)$ for $j=2,\dots,k$ and the embedding~$\mathcal E(x,Y)$ is obtained by merging the embeddings~$\mathcal E_{\mu_1}(v_1,Y)$ and~$\mathcal E_{\mu_j}(v_{j-1},v_{j})$ for~$j=2,\dots,k$ so that~$u$ and~$v$ lie on the outer face.

\smallskip\noindent\textbf{Case 4: $\mu$ is an R-node.} 
If~$s\notin\pert(\mu)$, then $P_{\mu}(v,u)$
must traverse every vertex in $\pert(\mu)$, starting with the edge $e$ counterclockwise following the parent edge, with all other edges of $\pert(\mu)$ to the left of $P_{\mu}(v,u)$. Since $v$, $u$, and $e$ lie on a common face, $P_{\mu}(v,u)$ follows only this face, so $\skel(\mu)$ is outerplanar; a contradiction, as the skeleton of an R-node is triconnected.

    \begin{figure}[t]
        \centering
        \begin{subfigure}[b]{.32\textwidth}
            \centering
            \includegraphics[page=1]{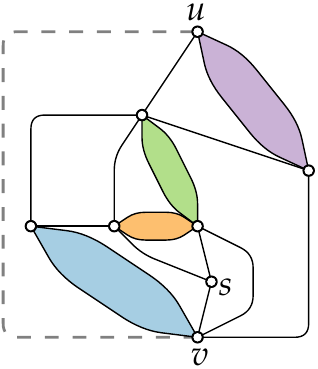}
            \caption{$\skel(\mu)$}
        \label{fig:r-node-pole-1}
        \end{subfigure}
        \hfill
        \begin{subfigure}[b]{.32\textwidth}
            \centering
            \includegraphics[page=3]{figures/r-node}
            \caption{The auxiliary graph}
        \label{fig:r-node-pole-2}
        \end{subfigure}
        \hfill
        \begin{subfigure}[b]{.32\textwidth}
            \centering
            \includegraphics[page=4]{figures/r-node}
            \caption{$P_\mu(u,\emptyset)$}
        \label{fig:r-node-pole-3}
        \end{subfigure}
        \caption{An R-node $\mu$ for the case that $s$ is a vertex of $\skel(\mu)$.}
        \label{fig:r-node-pole}
    \end{figure}

    Now, consider the case that $s\in\pert(\mu)$.  We start with the case that~$s$ is a vertex of~$\skel(\mu)$; see \cref{fig:r-node-pole}.  The path $P_{\mu}(u,Y)$ certainly must traverse the pertinent graphs of all children that are not Q-nodes and possibly also some of the Q-nodes.  To model this, we consider the auxiliary plane graph obtained from~$\skel(\mu)$ by replacing each virtual edge that corresponds to a non-Q-node child by a path of length~2.  We now employ the algorithm from \cref{thm:plane-staq-ham} for both embeddings of the auxiliary graph. We try every edge incident to~$s$ as a possible starting edge and check when we arrive at~$u$ whether all vertices except the vertices in~$Y$ have been visited. If this is successful, let~$v_1,\dots,v_\ell$ be the corresponding path in~$\skel(\mu)$ and let~$\mu_i$ be the child corresponding to the virtual edge~$\{v_i,v_{i+1}\}$ for~$i=1,\dots,\ell-1$.  If further~$(v_i,v_{i+1}) \in L(\mu_i)$ for~$i=1,\dots,\ell-1$, then $(v,u) \in L(\mu)$.  In that case, $P_\mu(v,u)$ is obtained by concatenating~$P_{\mu_i}(v_i,v_{i+1})$ for~$i=1,\dots,\ell-1$ and~$\mathcal E_\mu(v,u)$ is obtained from the embedding of the auxiliary graph by replacing each path of length~2 that represents a non-Q-node child~$\mu_i$ by~$\mathcal E_{\mu_i}(v_i,v_{i+1})$.
    If this test is not successful we repeat the above steps with the flipped embedding of the auxiliary graph.

    \begin{figure}[t]
        \centering
            \includegraphics[page=5]{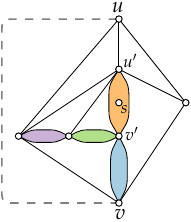}
        \hfill
            \includegraphics[page=7]{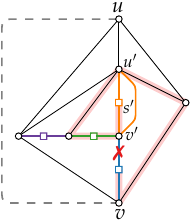}
        \hfill
            \includegraphics[page=8]{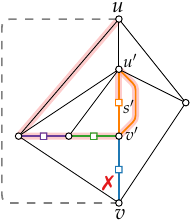}
        
        \medskip
            \includegraphics[page=10]{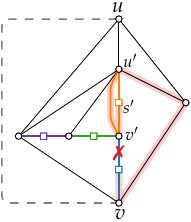}
        \hfill
            \includegraphics[page=11]{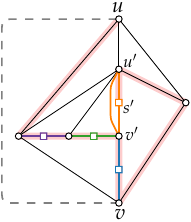}
        \hfill
            \includegraphics[page=12]{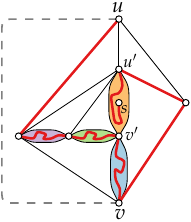}
        \caption{An R-node $\mu$ for the case that $s$ lies in a child $\nu$. }
        \label{fig:r-node-child}
    \end{figure}
    
    Otherwise~$s$ is contained in a child~$\nu$ of~$\mu$ with poles~$u',v'$; see \cref{fig:r-node-child}.  We consider the same auxiliary graph $H$ as above.  Let $s'$ be the vertex on the length-2 path between $u'$ and $v'$ in $H$.  We add the edge~$(u',v')$ to $H$ embedded either to the left or to the right of the path $\langle u',s',v'\rangle$; this way we obtain two different embeddings of the resulting graph. 
    We now employ the algorithm from \cref{thm:plane-staq-ham} for both embeddings.  Again, we try both starting edges incident to $s$ and for each of them, we check when we arrive at~$x$ whether all vertices except possibly the vertices in~$Y$ have been visited.
    This way, we obtain up to four solutions, depending on the starting edge and whether we use the edge $(u',v')$ or not. Let $x'\in\{u',v'\}$ such that $(s,x')$ is the starting edge of one such solution. If the path uses the edge $(u',v')$, then we have to check whether $(x',\emptyset)\in L(\nu)$; otherwise, we have to check whether $(x',\{u',v'\}\setminus\{x'\})\in L(\nu)$. 
    If the check is successful, then we compute the corresponding path $P_\mu(v,u)$ and embedding $\mathcal E_\mu(v,u)$ as in the case $s\notin\pert(\mu)$.
    This finishes the description of the R-node.

We conclude by mentioning that the running time stems from the fact that in an R-node that contains~$s$, we try~$O(n)$ starting edges for the path, where each try takes~$O(n)$ time.  Therefore, for a fixed pair of endvertices~$s,t$ testing the existence of an embedding that is $st$-sided takes~$O(n^2)$ time.  Since there are~$O(n^2)$ pairs of endvertices to try, the overall running time is~$O(n^4)$.
\end{proof}

\section{The \Rique-number of Complete Graphs}

In this section, we provide bounds on the density of graphs admitting $k$-page $\SQ$-layouts and on the $\rique$-number of complete graphs.

\begin{theorem}\label{thm:density}
    Any graph $G$ that admits a $k$-page $\SQ$-layout cannot have more than
    $(2n+2)k-k^2+(n-3)$ edges.
\end{theorem}
\begin{proof}
Let $v_1,\ldots,v_n$ be the linear order of the vertices and let $E_1,\ldots,E_k$ be the pages of a $k$-page $\SQ$-layout of $G$. 
Since, by \cref{thm:characterization}, each page is a planar graph, it has at most $3n-6$ edges. 
Since, however, the $n-1$ so-called \emph{spine edges} $(v_i,v_{i+1})$, $ i=1,\ldots,n-1$ can be added as head-edges to every page, every page has at most $2n-5$ non-spine edges. Next, we argue that there exists a $k$-page $\SQ$-layout $E'_1,\ldots,E'_k$ of $G$ such that each vertex $v_i$, $1\le i\le k$ contains edges only on pages $E'_1,\ldots,E'_i$. We start with $E'_i=E_i$, for each $1\le i\le k$.

\begin{figure}[t]
    \centering
    \includegraphics{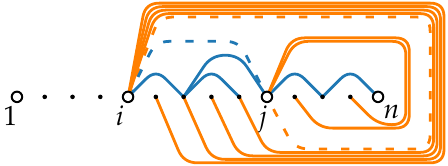}
    \caption{Illustration of page $i$ in the proof of \cref{thm:density}.}
    \label{fig:lowerbound}
\end{figure}

For $1\le i\le k$, assume that the first $i-1$ vertices $v_1,\ldots,v_{i-1}$ only have edges in $E'_1,\ldots,E'_{i-1}$ and consider the next vertex $v_i$ (see \cref{fig:lowerbound}). 
If $v_i$ also only has edges in $E'_1,\ldots,E'_{i-1}$, then the claim follows. Otherwise, let $(v_i,v_j)$, $i+1\le j\le n$ be the edge with $j$ maximal that does not lie in $E'_1,\ldots,E'_{i-1}$ and assume w.l.o.g.\ that $(v_i,v_j)\in E'_i$. 
By our assumption, there is no edge that stems from $v_1,\ldots,v_{i-1}$. Further, the edge $(v_i,v_j)$ blocks any possible tail-edge between two vertices in $v_{i+1},\ldots,v_{j-1}$ in $E'_i$. Hence, all tail-edges that end in a vertex in $v_{i+1},\ldots,v_{j-1}$ in $E'_i$ stem from $v_i$. Thus, we can add all edges from $v_i$ to $v_{i+1},\ldots,v_{j-1}$ to $E'_i$ as tail-edges. Since all edges from $v_1,\ldots,v_{i-1}$ to $v_i$ lie in $E'_1,\ldots, E_{i-1}'$, by the choice of $j$, so do all edges from $v_i$ to $v_{j+1},\ldots, v_n$. Thus, $E'_{i+1},\ldots,E'_k$ contain no edge of $v_i$.
Since any page $E'_i$, $1\le i\le k$ contains edges of at most $n-i+1$ vertices, it has at most $2(n-i+1)-5=2n-2i-3$ non-spine edges. Hence, the number of edges in $E'_1,\ldots,E'_k$ is at most
\[
n-1+\sum_{i=1}^k (2n-2i-3) = (2n-4)k-k^2+(n-1).\qedhere
\]
\end{proof}

\noindent We are now ready to present our bounds on the \rique-number of $K_n$.

\begin{theorem}\label{thm:staq-bounds}
  $0.2929(n-2)\approx (1-\frac{1}{\sqrt 2})(n-2) \le \sq(K_n) \le \lceil n/3 \rceil \approx 0.3333n$
\end{theorem}
\begin{proof} 
Let $k=\sq(K_n)$. As $K_n$ has $n(n-1)/2$ edges, \cref{thm:density} implies:
\[
(2n-4)k-k^2+(n-1) \ge \frac{n(n-1)}{2} 
\Leftrightarrow  k^2-(2n-4)k+(\frac{n^2}{2}-\frac{3n}{2}+1) \le 0
\]
The inequality above then gives the claimed lower bound as follows:
\[
k \ge n-2 -\frac{\sqrt{2}}{2} \sqrt{(n-2)(n-3)}
\ge n-2-\frac{\sqrt{2}(n-2)}{2}
=(1-\frac{1}{\sqrt{2}})(n-2)
\]

\begin{figure}[t]
    \centering
    \includegraphics{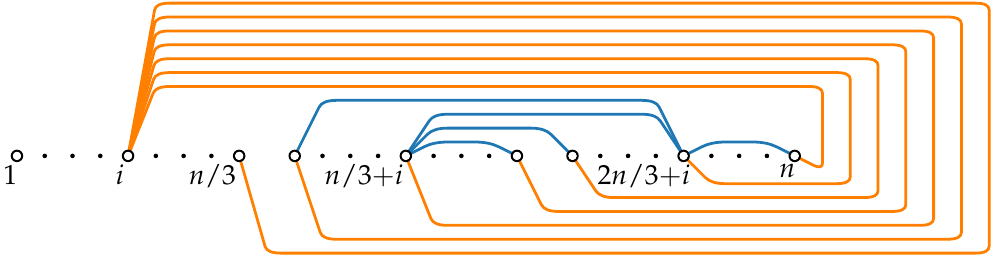}
    \caption{Illustration of page~$i$ in the upper bound of \cref{thm:staq-bounds}.}
    \label{fig:kn}
\end{figure}

\noindent We now show how to compute a layout of $K_n$ with $\lceil n/3 \rceil$ pages.  Assume w.l.o.g.\ that $n$ is divisible by~$3$.  Take an arbitrary stack layout of the clique on vertices~$v_{n/3+1},\dots,v_{v_n}$ on $n/3$ pages~\cite{DBLP:journals/jct/BernhartK79}.  Then put on page $i$ all edges of vertex~$v_i$ as tail-edges; see \cref{fig:kn}.
\end{proof}

\noindent We conclude this section with a few more insights on the \rique-number of complete graphs, which we derived by adjusting a formulation of the book embedding problem as a SAT instance~\cite{DBLP:conf/gd/Bekos0Z15}; for details see \cref{app:recognition}. This adjustment allowed us to obtain bounds on the \rique-number of $K_n$ for values of $n$ in $[4,\ldots,27]$; see \cref{tbl:kn} and \cref{fig:k7,fig:k11} for page-minimal layouts of $K_7$ and $K_{11}$.
\begin{toappendix}
\section{Omitted Material from \cref{sec:recognition}}
\label{app:recognition}
In our formulation, there exist three different types of variables, denoted by $\sigma$, $\phi$ and $\chi$, with the following meanings: 
\begin{enumerate*}[label=(\roman*)]
\item for a pair of vertices $u$ and $v$, variable $\sigma(u,v)$ is $\texttt{true}$, if and only if $u$ is to the left of $v$ along
the spine, 
\item for an edge $e$ and a page $i$, variable $\phi_i(e)$ is $\texttt{true}$, if and only if edge $e$ is assigned to page $i$ of the book, and 
\item for a pair of edges $e$ and $e'$, variable $\chi(e,e')$ is $\texttt{true}$, if and only if $e$ and $e'$ are assigned to the same page.
\end{enumerate*}  
Hence, there exist in total $O(n^2+m^2+pm)$ variables, where $n$ denotes the number of vertices of the graph, $m$ its number of edges, and $p$ the number of available pages. A set of $O(n^3+m^2)$ clauses ensures that the underlying order is indeed linear, and that no two edges of the same page cross; for details we point the reader to~\cite{DBLP:conf/gd/Bekos0Z15}. In our formulation, we neglected the clauses ensuring that edges of the same page do not cross, and we introduced clauses ensuring the absence of Pattern~\ref{pattern-staq} of \cref{lem:pattern-staq}. In particular, for every triplet of edges $(a,a')$, $(b,b')$ and $(c,c')$, we can guarantee that they do not form Pattern~\ref{pattern-staq} by the following clause: 
\begin{align*}
&\sigma(a,b) \wedge \sigma(b,c) \wedge \sigma(c,b') \wedge \sigma(b',a') \wedge \sigma(b',c')\rightarrow\\
& \neg(\chi((a,a'),(b,b'))\wedge\chi((b,b'),(c,c'))\wedge\chi((a,a'),(c,c')))
\end{align*}
With the aforementioned formulation we were able to derive bounds on the \rique-number of $K_n$ for values of $n$ in $[4,\ldots,27]$; see \cref{tbl:kn} for an overview and \cref{fig:k7,fig:k11} for page-minimal drawings of $K_7$ and $K_{11}$. For $K_{22}$ and $K_{25}$, we could not find the exact number and stopped the computation after 1 week.
\end{toappendix}

\begin{figure}[t]
    \centering
    \includegraphics[page=2]{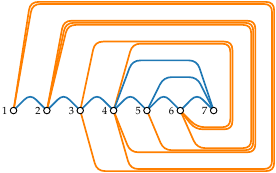}
    \hfill
    \includegraphics[page=3]{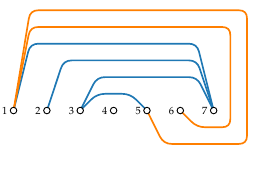}
    \caption{A 2-page $\SQ$-layout of $K_7$}
    \label{fig:k7}
\end{figure}

\begin{figure}[t]
    \centering
    \includegraphics[page=2,width=.32\linewidth]{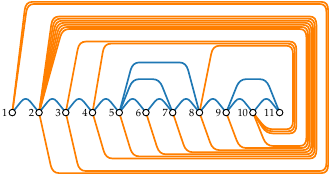}
    \hfill
    \includegraphics[page=3,width=.32\linewidth]{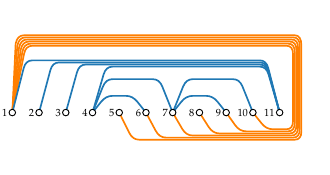}
    \hfill
    \includegraphics[page=4,width=.32\linewidth]{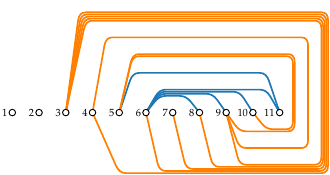}
    \caption{A 3-page $\SQ$-layout of $K_{11}$}
    \label{fig:k11}
\end{figure}

\begin{table}[t]
\caption{A summary of our results on the \rique-number of $K_n$}
\centering
\begin{tabular}{c@{\quad}|@{\quad}c@{\quad}c@{\quad}c@{\quad}c@{\quad}c@{\quad}c@{\quad}c@{\quad}c@{\quad}c@{\quad}c}
$n$        & 4 & 5--7 & 8--11 & 12--14 & 15--17 & 18--21 & 22  & 23--24 & 25  & 26--28 \\\hline
$\sq(K_n)$ & 1 & 2    & 3     & 4      & 5      & 6      & 6 or 7 & 7      & 7 or 8 & 8     
\end{tabular}
\label{tbl:kn}
\end{table}

\section{Conclusions and Open Problems}
In this work, we continued the study of linear layouts of graphs in relation to known data structures, in particular, in relation to the restricted-input deque. Several problems are raised by our work: 
\begin{enumerate*}[label=(\roman*)]
\item the most important one is the complexity of the recognition of graphs with \rique-number~$1$,
\item another quite natural problem is to further narrow the gap between our lower and upper bounds on the \rique-number of $K_n$; our experimental results indicate that there exist room for improvement in the upper bound, 
\item for complete bipartite graphs, we did not manage to obtain improved bounds (besides the obvious ones that one may derive from their stack- or queue-number), 
\item another interesting question regards the \rique-number of planar graphs, which ranges between $2$ and $4$ (i.e., the upper bound by their stack-number); the same problem can be studied also for subclasses of planar graphs (e.g., planar $3$-trees). 
\end{enumerate*}

\bibliographystyle{splncs04}
\bibliography{stacks,queues,general}
\end{document}